\documentclass[preprint]{elsarticle}
\usepackage[utf8]{inputenc}

\usepackage{amsmath}
\usepackage{amssymb}
\usepackage{amsthm}
\usepackage{stmaryrd}

\usepackage{xspace}

\usepackage{tabularx}
\usepackage{booktabs}

\newtheorem{definition}{Definition}
\newtheorem{theorem}{Theorem}
\newtheorem{example}{Example}
\newtheorem{corollary}{Corollary}

\newcommand{\expspace}{\ensuremath{\mathbf{EXPSPACE}}\xspace}
\newcommand{\pspace}{\ensuremath{\mathbf{PSPACE}}}
\newcommand{\np}{\ensuremath{\mathbf{NP}}}
\newcommand{\nexptime}{\ensuremath{\mathbf{NEXPTIME}}}
\newcommand{\nnexptime}{\ensuremath{\mathbf{2{-}NEXPTIME}}}

\newcommand{\atime}[2]{\ensuremath{\mathbf{ATIME}(#1, #2)}}
\newcommand{\aexpp}{\ensuremath{\mathbf{AEXP}(\mathrm{poly})}\xspace}
\newcommand{\atm}{\ensuremath{\textsc{atm}}\xspace}
\newcommand{\bvv}{\ensuremath{\mathsf{BV2}}\xspace}
\newcommand{\bv}{\ensuremath{\mathsf{BV}}\xspace}
\newcommand{\ar}[1]{\ensuremath{\mathrm{ar}(#1)}\xspace}
\newcommand{\extract}{\ensuremath{\mathit{extract}}}

\newcommand{\Oh}{\ensuremath{\mathcal{O}}}
\newcommand{\N}{\ensuremath{\mathbb{N}}}
\newcommand{\F}{\ensuremath{\mathcal{F}}}
\newcommand{\I}{\ensuremath{\mathcal{I}}}
\newcommand{\SO}{\ensuremath{\mathsf{SO}_2}}
\newcommand{\eval}[1]{\ensuremath{\llbracket #1 \rrbracket}}

\newcommand{\smt}{\textsc{smt}\xspace}

\usepackage[dvipsnames]{xcolor}
\usepackage{hyperref}

\title{On the Complexity of the Quantified \\ Bit-Vector Arithmetic with
  Binary Encoding}

\author[fi]{M.~Jonáš\corref{cor1}}
\ead{martin.jonas@mail.muni.cz}

\author[fi]{J.~Strejček}
\ead{strejcek@mail.muni.cz}

\cortext[cor1]{Corresponding author}
\address[fi]{Faculty of Informatics, Masaryk University \\ Botanická
  68a, 602\,00, Brno, Czech Republic}

\begin{document}

\begin{abstract}
  We study the precise computational complexity of deciding
  satisfiability of first-order quantified formulas over the theory of
  fixed-size bit-vectors with binary-encoded bit-widths and constants.
  This problem is known to be in \expspace and to be \nexptime-hard.
  We show that this problem is complete for the complexity class
  \aexpp~-- the class of problems decidable by an alternating Turing
  machine using exponential time, but only a polynomial number of
  alternations between existential and universal states.
\end{abstract}


\maketitle

\section{Introduction}

The first-order theory of fixed-size bit-vectors is widely used for
describing properties of software and hardware. Although most current
applications use only the quantifier-free fragment of this logic, there
are several use cases that benefit from using bit-vector formulas
containing
quantifiers~\cite{GSV09,SGF10,CKRW13,KLW13,MBLB16}. Consequently,
computational complexity of quantified bit-vector logic has been
investigated in recent years. It has been shown that deciding
satisfiability of quantified bit-vector formulas is $\pspace$-complete
and it becomes $\nexptime$-complete when uninterpreted functions are
allowed in addition to quantifiers~\cite{WHM13}.

However, these results suppose that all scalars in the formula are
represented in the unary encoding, which is not the case in practice,
because in most of real-world applications, bit-widths and constants
are encoded logarithmically. For example, the format
\textsc{smt-lib}~\cite{BFT15}, which is an input format for most of
the state-of-the-art \smt solvers, represents all scalar values as
decimal numbers. Such representation can be exponentially more
succinct than the representation using unary-encoded scalars. The
satisfiability problem for bit-vector formulas with binary-encoded
scalars has been recently investigated by Kovásznai et
al.~\cite{KFB16}. They have shown that the satisfiability of
quantified bit-vector formulas with binary-encoded scalars and with
uninterpreted functions is $\nnexptime$-complete. The situation for
the same problem without uninterpreted functions is not so clear:
deciding satisfiability of quantified bit-vector formulas with binary
encoded scalars and without uninterpreted functions (we denote this
problem as $\bvv$ satisfiability) is known to be in $\expspace$ and to
be $\nexptime$-hard, but its precise complexity has remained
unknown~\cite{KFB16}.

In this paper, we solve this open problem by identifying the
complexity class for which $\bvv$ satisfiability is complete. We use
the notion of an alternating Turing machine introduced by Chandra et
al.~\cite{CKS81} and show that the $\bvv$ satisfiability problem is
complete for the class $\aexpp$ of problems solvable by an alternating
Turing machine using exponential time, but only a polynomial number of
alternations.

\section{Quantified Bit-Vector Formulas}

The \emph{theory of fixed-size bit-vectors} (\emph{\bv} or
\emph{bit-vector theory} for short) is a many-sorted first-order
theory with infinitely many sorts corresponding to bit-vectors of
various lengths. Each bit-vector variable has an explicitly assigned
sort, e.g.~$x^{[3]}$ is a bit-vector variable of bit-width $3$. The
\bv theory uses only three predicates, namely equality ($=$), unsigned
inequality of binary-encoded non-negative integers ($\leq_u$), and
signed inequality of integers in $2$'s complement representation
($\leq_s$). The signature also contains constants $c^{[n]}$ for each
$n \geq 1$ and $0 \leq c \leq 2^n - 1$, and various interpreted
functions, namely addition ($+$), multiplication ($*$), unsigned
division ($\div$), bitwise negation ($\sim$), bitwise and ($\&$),
bitwise or ($|$), bitwise exclusive or ($\oplus$), left-shift
($\ll$), right-shift ($\gg$), concatenation ($\cdot$), and extraction
of a subword starting at the position $i$ and ending at the position
$j$ ($\extract(\_,i,j)$). Although various sources define the full
bit-vector theory with different sets of functions, all such
definitions can be polynomially reduced to each other~\cite{KFB16}.
All numbers occurring in the formula, i.e.~values of constants,
bit-widths and bounds $i,j$ of extraction, are called \emph{scalars}.

\begin{table}[!btp]
  \begin{center}
    \setlength{\tabcolsep}{0.6em}
    \begin{tabular}{l l l}
      \toprule
      & Expression & Size \\
      \midrule
      Constant & $|c^{[n]}|$ & $L(c) + L(n)$ \\
      Variable & $|x^{[n]}|$ & $1 + L(n)$ \\
      Operation & $|o(t_1, \ldots, t_k, i_1, \ldots, i_p)|$ & $1 + \sum_{1 \leq i \leq k} |t_i| + \sum_{1 \leq j \leq p} L(i_j)$ \\
      Quantifier & $|Qx^{[n]} \varphi|$ & $|x^{[n]}| + |\varphi|$ \\
      \bottomrule
    \end{tabular}
  \end{center}
  \caption{Recursive definition of the formula size. Operations
    include logical connectives, function symbols, and predicate
    symbols. Each $t_i$ denotes a subterm or a subformula, each $i_j$
    denotes a scalar argument of an operation, and
    $Q \in \{ \exists, \forall \}$~\cite{KFB16}.}
  \label{tbl:formulaSize}
\end{table}

There are more ways to encode scalars occurring in the bit-vector
formula: in the \emph{unary encoding} or in a \emph{logarithmic
  encoding}. In this paper, we focus only on formulas using the
\emph{binary encoding}. This covers all logarithmic encodings, since
all of them are polynomially reducible to each other. In the binary
encoding, $L(n)$ bits are needed to express the number $n$, where
$L(0) = 1$ and $L(n) = \lfloor \log_2 n \rfloor + 1$ for all $n > 0$.
The entire formula is encoded in the following way: each constant
$c^{[n]}$ has both its value $c$ and bit-width $n$ encoded in binary,
each variable $x^{[n]}$ has its bit-width $n$ encoded in binary, and
all scalar arguments of functions are encoded in binary. The size of
the formula $\varphi$ is denoted $|\varphi|$. The recursive definition
of $|\varphi|$ is given in Table~\ref{tbl:formulaSize}. For quantified
formulas with binary-encoded scalars, we define the corresponding
satisfiability problem:

\begin{definition}[\cite{KFB16}]
  The \emph{$\bvv$ satisfiability problem} is to decide satisfiability
  of a given closed quantified bit-vector formula with all scalars
  encoded in binary.
\end{definition}

Similarly to Kovásznai et al.~\cite{KFB16}, we use an \emph{indexing}
operation, which is a special case of the extraction operation that
produces only a single bit. In particular, for a term $t^{[n]}$ and a
number $0 \leq i < n$, the indexing operation $t^{[n]}[i]$ is defined
as $\extract(t^{[n]},i,i)$. We assume that bits of bit-vectors are
indexed from the least significant. For example, given a bit-vector
variable $x^{[6]}=x_5x_4x_3x_2x_1x_0$, the value of $x^{[6]}[1]$
refers to $x_1$. In the following, we use a more general version of
the indexing operation, in which the index can be an arbitrary
bit-vector term, not only a fixed scalar. This operation can be
defined using the indexing operation and the bit-shift operation with
only a linear increase in the size of the term:
\[
  t^{[n]}[s^{[n]}]~\stackrel{\mathrm{df}}{\equiv}~(t^{[n]} \gg s^{[n]})[0].
\]

\section{Alternation Complexity}

We assume a basic familiarity with an \emph{alternating Turing
  machine}~(\atm) introduced by Chandra, Kozen, and
Stockmeyer~\cite{CKS81}, and basic concepts from the complexity
theory, which can be found for example in Kozen~\cite{Koz06}. We
recall that each state of an \atm is either \emph{existential} or
\emph{universal}. Existential states behave like states of a
non-deterministic Turing machine: a run passing through an existential
state continues with one of the possible successors. In contrast to this,
a run entering a universal state forks and continues into all possible
successors. Hence, runs of an \atm are trees. Such a run is accepting
if each branch of the run ends in an accepting state.

This section recalls some complexity classes related to alternating
Turing machines. Computations in such complexity classes are bounded
not only by time and memory, but also by the number of alternations
between existential and universal states during the computation.
Although bounding both time and memory is useful in some applications,
in this paper we need only complexity classes related to {\atm}s that
are bounded in time and the number of alternations. Therefore, the
following definition introduces a family of complexity classes
parameterized by the number of steps and alternations used by
corresponding {\atm}s.

\begin{definition}
  Let $t,g \colon \N \rightarrow \N$ be functions such that
  $g(n) \geq 1$.  We define the complexity class $\atime{t}{g}$ as the
  class of all problems $A$ for which there is an alternating Turing
  machine that decides $A$ and, for each input of length $n$, it needs
  at most $t(n)$ steps and $g(n) - 1$ alternations along every
  branch of every run. If $T$ and $G$ are classes of functions, let
  $\atime{T}{G} = \bigcup_{t \in T, g \in G}\atime{t}{g}$.
\end{definition}

Chandra et al.~have observed several relationships between classical
complexity classes related to time and memory and the complexity
classes defined by {\atm}s~\cite{CKS81}. We recall relationships
between alternating complexity classes and the classes $\nexptime$ and
$\expspace$, which are important for this paper. It can easily be seen
that the class $\nexptime$ corresponds to all problems solvable by an
alternating Turing machine that starts in an existential state and can
use exponential time and no alternations: this yields an inclusion
$\nexptime \subseteq \atime{2^{\Oh(n)}}{1}$.  On the other hand,
results of Chandra et al.~imply that $\expspace$ is precisely the
complexity class $\atime{2^{n^{\Oh(1)}}}{2^{n^{\Oh(1)}}}$ of problems
solvable in exponential time and with exponential number of
alternations. An interesting class that lies in between those two
complexity classes can be obtained by bounding the number of steps
exponentially and the number of alternations polynomially. This class
is called $\aexpp$.

\begin{definition}
  $\aexpp \stackrel{\mathrm{df}}{=} \atime{2^{n^{\Oh(1)}}}{n^{\Oh(1)}}.$
\end{definition}

The following inclusions immediately follow from the mentioned results.
$$\nexptime \subseteq \aexpp \subseteq \expspace$$ However, it is
unknown whether any of the inclusions is strict.

\section{Complexity of $\bvv$ Satisfiability}

In this section, we show that the $\bvv$ satisfiability problem is
$\aexpp$-complete.  First, we prove that the problem is in the class
$\aexpp$.

\begin{theorem}
  The $\bvv$ satisfiability problem is in $\aexpp$.
\end{theorem}
\begin{proof}
  We describe the alternating Turing machine solving the problem. For
  a given $\bvv$ formula $\varphi$, the machine first converts the
  formula to the prenex normal form, which can be done in polynomial
  time without any alternations~\cite{Har09}. The machine then assigns
  values to all existentially quantified variables using existential
  states and to all universally quantified variables using universal
  states. Although this requires exponential time, as there are
  exponentially many bits whose value has to be assigned, only a
  polynomial number of alternations is required, because the formula
  $\varphi$ can contain only polynomially many quantifiers.

  Finally, the machine uses the assignment to evaluate the
  quantifier-free part of the formula. If the result of the evaluation
  is true, the machine accepts; it rejects otherwise. The evaluation
  takes exponential time and no quantifier alternations: the machine
  replaces all variables by exponentially many previously assigned
  bits and computes results of all operations from the bottom of the
  syntactic tree of the formula up. The computation of each of the
  operations takes time polynomial in the number of bits, which is
  exponential.
\end{proof}

In the rest of this section, we show that the $\bvv$ satisfiability
problem is also $\aexpp$-hard. In particular, we present a reduction
of a known $\aexpp$-hard \emph{second-order Boolean formulas
  satisfiability problem}~\cite{Loh12,Luc16} to the $\bvv$
satisfiability.

Intuitively, the \emph{second-order Boolean logic} ($\SO$) can be obtained
from a quantified Boolean logic by adding function symbols and
quantification over such symbols. Alternatively, the $\SO$ logic
corresponds to the second-order predicate logic restricted to the
domain $\{ 0, 1 \}$. Lohrey and Lück have shown that by bounding the
number of quantifier alternations in second-order Boolean formulas,
problems complete for all levels of the exponential hierarchy can be
obtained. Moreover, if the number of quantifier alternations is
unbounded, the problem of deciding satisfiability of quantified
second-order Boolean formulas is $\aexpp$-complete~\cite{Loh12,Luc16}.

We now introduce the $\SO$ logic more formally. The definitions of the
syntax and semantics of $\SO$ used in this paper are due to Hannula et
al.~\cite{HKLV16}.

\begin{definition}[$\SO$ syntax~\cite{HKLV16}]
  Let $\mathcal{F}$ be a countable set of function symbols, where each
  symbol $f \in \mathcal{F}$ is given an arity
  $\ar{f} \in \mathbb{N}_0$. The set $\SO(\F)$ of \emph{quantified
    Boolean second-order formulas} is defined inductively as
  \[
    \varphi ::= \varphi \wedge \varphi \mid \neg \varphi \mid \exists
    f \varphi \mid \forall f \varphi \mid f(\underbrace{\varphi,
      \ldots, \varphi}_{\ar{f} \text{ times }}),
  \]
  where $f \in \mathcal{F}$.
\end{definition}

\begin{definition}[$\SO$ semantics~\cite{HKLV16}]
  An \emph{$\F$-interpretation} is a function $\I$ that assigns to
  each symbol $f \in \F$ a Boolean function of the corresponding
  arity, i.e.~$\I(f) \colon \{0,1\}^{\ar{f}} \rightarrow \{0,1\}$ for
  each $f \in \F$. The valuation of a formula $\varphi \in \SO(\F)$ in
  $\I$, written $\eval{\varphi}_\I$, is defined recursively as
  \begin{alignat*}{3}
    &\eval{\varphi \wedge \psi}_\I &&= \eval{\varphi}_\I * \eval{\psi}_\I, \\
    &\eval{\neg \varphi}_\I &&= 1 - \eval{\varphi}_\I, \\
    &\eval{f(\varphi_1, \ldots, \varphi_n)}_\I &&= \I(f)(\eval{\varphi_1}_\I, \ldots, \eval{\varphi_n}_\I), \\
    &\eval{\exists f \varphi}_\I &&= \max \left \{ \eval{\varphi}_{\I[f \mapsto F]} \mid F \colon \{ 0,1 \}^{\ar{f}} \rightarrow \{0,1\} \right \}, \\
    &\eval{\forall f \varphi}_\I &&= \min \left \{ \eval{\varphi}_{\I[f \mapsto F]} \mid F \colon \{ 0,1 \}^{\ar{f}} \rightarrow \{0,1\} \right \},
  \end{alignat*}
  where $\I [f \mapsto F]$ is the function defined as
  $\I[f \mapsto F](f) = F$ and $\I[f \mapsto F](g) = \I(g)$ for all
  $g \not = f$.

  An $\SO$ formula $\varphi$ is \emph{satisfiable} if
  $\eval{\varphi}_\I=1$ for some $\I$. 
\end{definition}

We call function symbols of arity $0$ \emph{propositions} and all
other function symbols \emph{proper functions}. An $\SO$ formula
$\varphi$ is in the \emph{prenex normal form} if it has the form
$\overline{Q} \psi$, where $\overline{Q}$ is a sequence of quantifiers
called a \emph{quantifier prefix}, $\psi$ is a quantifier-free formula
called a \emph{matrix}, and all proper functions are quantified before
propositions. In the following, we fix an arbitrary countable set of
function symbols $\mathcal{F}$ and instead of $\SO(\F)$, we write only
$\SO$.

\begin{definition}
  The \emph{$\SO$ satisfiability problem} is to decide whether a given
  closed $\SO$ formula in the prenex normal form is satisfiable.
\end{definition}

\begin{theorem}[\cite{Loh12,Luc16}]
  The $\SO$ satisfiability problem is $\aexpp$-complete.
\end{theorem}

We now show a polynomial time reduction of $\SO$ satisfiability to $\bvv$
satisfiability and thus finish the main claim of this paper, which
states that the $\bvv$ satisfiability problem is $\aexpp$-complete.

\begin{theorem}
  \label{thm}
  The $\bvv$ satisfiability problem is $\aexpp$-hard.
\end{theorem}
\begin{proof}
  We present a polynomial time reduction of $\SO$ satisfiability to
  $\bvv$ satisfiability. Let $\varphi$ be an $\SO$ formula with a
  quantifier prefix $\overline{Q}$ and a matrix $\psi$,
  i.e.~$\varphi = \overline{Q} \psi$ where $\psi$ is a quantifier-free
  formula. We construct a bit-vector formula $\varphi^{BV}$, such that
  $\varphi$ is satisfiable iff the formula $\varphi^{BV}$ is
  satisfiable.

  In the formula $\varphi^{BV}$, each function symbol $f$ of the
  formula $\varphi$ is represented by a bit-vector variable $x_f$ of
  bit-width $2^{\ar{f}}$. Intuitively, the bits of the variable $x_f$
  will encode values $f(b_{n-1}, \ldots, b_{0})$ for all possible
  inputs $b_0, \ldots, b_{n-1} \in \{ 0, 1 \}$. In particular, the
  value $f(b_{n-1}, \ldots, b_{0})$ is represented as the bit on the
  index $\sum_{i=0}^{n-1} (2^ib_i)$ in the bit-vector $x_f$.
  Equivalently, this index can be expressed as the numerical value of
  the bit-vector $b_{n-1} b_{n-2} \ldots b_0$. For example, for a
  ternary function symbol $f$, bits of the bit-vector value
  $x_f = x_7x_6x_5x_4x_3x_2x_1x_0$ will represent values $f(1,1,1)$,
  $f(1,1,0)$, $f(1,0,1)$, $f(1,0,0)$, $f(0,1,1)$, $f(0,1,0)$,
  $f(0,0,1)$, and $f(0,0,0)$, respectively.

  The reduction proceeds in two steps. First, we inductively construct
  a bit-vector term $\psi^{BV}$ of bit-width $1$, which corresponds to
  the formula $\psi$:

  \begin{itemize}
  \item If $\psi \equiv \rho_1 \wedge \rho_2$, we set
    $\psi^{BV} \equiv \rho_1^{BV} \mathbin{\&} \rho_2^{BV}$.
  \item If $\psi \equiv \neg \rho$, we set
    $\psi^{BV} \equiv {\sim} \rho^{BV}$.
  \item If $\psi \equiv f()$ (i.e.~$f$ is a proposition), we set
    $\psi^{BV} \equiv x_f^{[1]}$.
  \item If $\psi \equiv f(\rho_{n-1}, \ldots, \rho_{0})$ where $n = \ar{f}$,
    we set
    \[
      \psi^{BV} \equiv x_f^{[2^n]} \left [0^{[2^n - n]} \cdot \rho_{n-1}^{BV}
      \cdot \rho_{n-2}^{BV} \cdot \ldots \cdot \rho_{0}^{BV} \right].
    \]
    Note that because both arguments of the indexing operation have to
    be of the same sort, $2^n - n$ additional bits have to be added to
    the index term to get a term of the same bit-width as the term
    $x_f^{[2^n]}$.
  \end{itemize}

  In the second step, we replace each quantifier $Q_i f$ in the
  quantifier prefix $\overline{Q}$ by a bit-vector quantifier
  $Q_i x_f^{[2^n]}$, where $n = \ar{f}$, and thus obtain a sequence of
  bit-vector quantifiers $\overline{Q}^{BV}$. The final formula
  $\varphi^{BV}$ is then $\overline{Q}^{BV} (\psi^{BV} = 1^{[1]})$.

  Due to the binary representation of the bit-widths, the formula
  $\varphi^{BV}$ is polynomial in the size of the formula~$\varphi$.
\end{proof}

\begin{example}
  Consider an $\SO$ formula
  \[
    \exists f \forall p \forall q \,.\, \neg f(p,p,q) \wedge f(p,q \wedge
    \neg q, q),
  \]
  where $f$ is a ternary function symbol and $p, q$ are
  propositions. Then the result of the described reduction is the
  formula
  \[
    \exists x_f^{[8]} \forall x_p^{[1]} \forall x_q^{[1]} \; ({\sim}
    x_f^{[8]}[0^{[5]} \cdot x_p \cdot x_p \cdot x_q] \mathbin{\&}
    x_f^{[8]}[0^{[5]} \cdot x_p \cdot (x_q \mathbin{\&} {\sim}x_q)
    \cdot x_q]~=~1^{[1]}).
  \]
\end{example}

\begin{corollary}
  The $\bvv$ satisfiability problem is $\aexpp$-complete.
\end{corollary}

\begin{table}[!tbp]
  \begin{center}
    \setlength{\tabcolsep}{0.39em}
    \begin{tabular}{l l l l r}
      \toprule
      & \multicolumn{4}{c}{Quantifiers} \\
      \cmidrule(l){2-5}
      & \multicolumn{2}{c}{No} & \multicolumn{2}{c}{Yes} \\
      \cmidrule(l){2-3}
      \cmidrule(l){4-5}
      & \multicolumn{2}{c}{Uninterpreted functions} & \multicolumn{2}{c}{Uninterpreted functions} \\
      Encoding & \multicolumn{1}{c}{No} & \multicolumn{1}{c}{Yes} & \multicolumn{1}{c}{No} & \multicolumn{1}{c}{Yes} \\
      \midrule
      Unary & \np & \np & \pspace & \nexptime \\
      Binary\hspace{.1em} & \nexptime & \nexptime & \aexpp & \nnexptime \\
      \bottomrule
    \end{tabular}
  \end{center}
  \caption{Completeness results for various bit-vector logics and
    encodings. This is the table presented by Fröhlich et
    al.~\cite{FKB13} extended by the result proved in this paper.}
  \label{tbl:complexity}
\end{table}

\section{Conclusions}

We have identified the precise complexity class of deciding
satisfiability of a quantified bit-vector formula with binary-encoded
bit-widths. This paper shows that the problem is complete for the
complexity class $\aexpp$, which is the class of all problems solvable
by an alternating Turing machine that can use exponential time and a
polynomial number of alternations. This result settles the open
question raised by Kovásznai et al.~\cite{KFB16}. Known completeness
results for various bit-vector logics including the result proven in
this paper are summarized in Table~\ref{tbl:complexity}.

\section*{Acknowledgements}
Authors of this work are supported by the Czech Science Foundation,
project No.~GBP202/12/G061.

\section*{References}

\bibliography{complexity}

\end{document}